\documentclass[aps,pra,twocolumn,a4]{revtex4-1}

\usepackage[inline]{enumitem}

\usepackage{graphicx}
\usepackage{mathptmx}
\usepackage{amssymb}
\usepackage{amsmath}
\usepackage{xfrac}
\usepackage{bbold}
\usepackage{setspace}
\usepackage{multirow}
\usepackage{amsthm}
\usepackage{array}

\newcommand{\ket}[1]{\left| #1 \right>}
\newcommand{\bra}[1]{\left< #1 \right|}
\newcommand{\braket}[2]{\left< #1 | #2 \right>}
\newcommand{\ketbra}[2]{\ket{#1} \bra{#2}}

\theoremstyle{plain}
\newtheorem{theorem}{Theorem}[section]
\newtheorem{proposition}[theorem]{Proposition}
\newtheorem{corollary}[theorem]{Corollary}

\theoremstyle{definition}
\newtheorem{definition}[theorem]{Definition}

\usepackage{tikz}
\usepackage{pgfplots}

\usetikzlibrary{trees}
\usetikzlibrary{topaths}
\usetikzlibrary{decorations.pathmorphing}
\usetikzlibrary{decorations.markings}
\usetikzlibrary{snakes,matrix,backgrounds,folding}
\usetikzlibrary{chains,scopes,positioning,fit}
\usetikzlibrary{arrows,shadows}
\usetikzlibrary{calc}
\usetikzlibrary{chains}
\usetikzlibrary{shapes,shapes.geometric,shapes.misc}

\pgfdeclarelayer{edgelayer}
\pgfdeclarelayer{nodelayer}
\pgfsetlayers{background,edgelayer,nodelayer,main}

\tikzstyle{every picture}=[baseline=-0.25em]
\tikzstyle{dotpic}=[scale=0.6]
\tikzstyle{dot graph}=[shorten <=-0.1mm,shorten >=-0.1mm,scale=0.6]
\tikzstyle{digraph}=[-latex]

\tikzstyle{none}=[inner sep=0mm]
\tikzstyle{every loop}=[]

\tikzstyle{SLClassical}=[regular polygon,regular polygon sides=3,draw=black,fill=white,scale=2]

\tikzstyle{SLQuantum}=[circle,fill=white,draw=black,scale=3]

\tikzstyle{dot}=[inner sep=0.5mm,fill=black,draw=black,shape=circle]
\tikzstyle{ddot}=[minimum size=1em,inner sep=-0.5mm,fill=black,draw=black,shape=circle]
\tikzstyle{white dot}=[dot,fill=white]
\tikzstyle{gray dot}=[dot,fill=gray!50]
\tikzstyle{box vertex}=[draw=black,rectangle]

\tikzstyle{green dot}=[dot,fill=green]
\tikzstyle{red dot}=[dot,fill=red]
\tikzstyle{green ddot}=[ddot,fill=green]
\tikzstyle{red ddot}=[ddot,fill=red]
\tikzstyle{bayesian dot}=[dot,fill=black]

\tikzstyle{square box}=[rectangle,fill=white,draw=black,minimum height=6mm,minimum width=6mm]
\tikzstyle{square gray box}=[rectangle,fill=gray!30,draw=black,minimum height=6mm,minimum width=6mm]
\tikzstyle{point}=[regular polygon,regular polygon sides=3,draw=black,scale=0.75,inner sep=-0.5pt]
\tikzstyle{copoint}=[point,regular polygon rotate=180]

\tikzstyle{open graph}=[baseline=-0.25em]
\tikzstyle{greybg}=[background rectangle/.style={fill=black!5,draw=black!30,rounded corners=1ex}, show background rectangle]

\tikzstyle{edge point}=[circle,minimum width=1mm,fill=black,inner sep=0mm]
\tikzstyle{vertex point}=[circle,minimum width=4mm,fill=white,draw=black,inner sep=0mm]
\tikzstyle{edge label}=[inner sep=2pt, font=\small]
\tikzstyle{on edge label}=[fill=white, font=\footnotesize, inner sep=1 pt]

\newcommand{\edgetick}{{\arrow[black,scale=0.5,thick]{|}}}
\tikzstyle{white edge}=[line width=5pt,white]
\tikzstyle{tick}=[postaction=decorate,decoration={markings, mark=at position 0.5 with \edgetick}]
\tikzstyle{map edge}=[|-latex,very thick, gray!40, shorten <=1mm, shorten >=0.5mm]

\newcommand{\edgearrow}{{\arrow[black]{>}}}
\newcommand{\edgearrowback}{{\arrow[black]{<}}}

\tikzstyle{diredge}=[postaction=decorate,decoration={markings, mark=at position 0.55 with \edgearrow}]
\tikzstyle{diredgeback}=[postaction=decorate,decoration={markings, mark=at position 0.55 with \edgearrowback}]

\tikzstyle{small dotpic}=[dotpic,scale=0.6]
\tikzstyle{small}=[inner sep=0.4mm]
\tikzstyle{small dot}=[small,dot]
\tikzstyle{small white dot}=[small dot,fill=white]
\tikzstyle{small gray dot}=[small dot,fill=gray!50]

\tikzstyle{cnot}=[fill=white,shape=circle,inner sep=-1.4pt]
\tikzstyle{bang box}=[draw=black,dashed,minimum height=12mm,minimum width=12mm,fill=gray!20]
\tikzstyle{wire label}=[font=\footnotesize, auto]

\usepackage{tikz-3dplot}
\usepackage{subfig}

\newcommand{\drawFrontCube}[1]{
    \draw[thick] (0,0,0) -- (#1,0,0) -- (#1,#1,0) ;
    \draw[thick] (0,0,#1) -- (#1,0,#1) -- (#1,#1,#1) -- (0,#1,#1) -- cycle;
    \draw[thick] (0,0,#1) -- (0,0,0);
    \draw[thick] (#1,0,#1) -- (#1,0,0);
    \draw[thick] (#1,#1,#1) -- (#1,#1,0);
}

\newcommand{\drawBackCube}[1]{
    \draw[thick] (#1,#1,0) -- (0,#1,0) --(0,0,0);
    \draw[thick] (0,#1,#1) -- (0,#1,0);
}

\newcommand{\pointAt}[4]{
    \draw[#4,fill=#4, opacity=0.5] (#1,#2,#3) -- (1+#1,#2,#3) -- (1+#1,1+#2,#3) -- (#1,1+#2,#3) -- cycle ;
    \draw[#4,fill=#4, opacity=0.5] (#1,1+#2,#3) -- (1+#1,1+#2,#3) -- (1+#1,1+#2,1+#3) -- (#1,1+#2,1+#3) -- cycle ;
    \draw[#4,fill=#4, opacity=0.5] (#1,#2,#3) -- (#1,1+#2,#3) -- (#1,1+#2,1+#3) -- (#1,#2,1+#3) -- cycle ;
    
    \draw[#4,fill=#4, opacity=0.5] (#1,#2,1+#3) -- (1+#1,#2,1+#3) -- (1+#1,1+#2,1+#3) -- (#1,1+#2,1+#3) -- cycle ;
    \draw[#4,fill=#4, opacity=0.5] (#1,#2,#3) -- (1+#1,#2,#3) -- (1+#1,#2,1+#3) -- (#1,#2,1+#3) -- cycle ;
    \draw[#4,fill=#4, opacity=0.5] (#1+1,#2,#3) -- (1+#1,1+#2,#3) -- (1+#1,1+#2,1+#3) -- (1+#1,#2,1+#3) -- cycle ;
}
  
  \newcommand{\drawInnerDivider}{
    \draw[dotted] (0,0,2) -- (4,0,2);
    \draw[dotted] (0,2,2) -- (4,2,2);
    \draw[dotted] (0,4,2) -- (4,4,2);
    
    \draw[dotted] (2,0,0) -- (2,4,0);
    \draw[dotted] (0,2,0) -- (4,2,0);
    \draw[dotted] (2,0,4) -- (2,4,4);
    \draw[dotted] (0,2,4) -- (4,2,4);

    \draw[dotted] (2,0,0) -- (2,0,4);
    \draw[dotted] (0,2,0) -- (0,2,4);
    \draw[dotted] (2,2,0) -- (2,2,4);
    \draw[dotted] (4,2,0) -- (4,2,4);
    \draw[dotted] (2,4,0) -- (2,4,4);
    
    \draw[dotted] (0,0,2) -- (0,4,2);
    \draw[dotted] (2,0,2) -- (2,4,2);
    \draw[dotted] (4,0,2) -- (4,4,2);
}
  
    \newcommand{\drawNKLCube}[2]{
 \tdplotsetmaincoords{55}{15}
       \begin{figure}
        \caption{#2}
\centering
  \begin{tikzpicture}[scale=1,tdplot_main_coords]
	\drawInnerDivider
	\drawBackCube{4}
	#1
	\drawFrontCube{4}
  \end{tikzpicture}
  \end{figure}
  
  }

   \newcommand{\drawSmallNKLCubes}[3]{
 \tdplotsetmaincoords{55}{15}
       \begin{figure}
        \caption{#3}
\centering
\begin{tabular}{cc}
  \begin{tikzpicture}[scale=.8,tdplot_main_coords]
	\drawInnerDivider
	\drawBackCube{4}
	#1
	\drawFrontCube{4}
  \end{tikzpicture}
  &
  \begin{tikzpicture}[scale=.8,tdplot_main_coords]
	\drawInnerDivider
	\drawBackCube{4}
	#2
	\drawFrontCube{4}
  \end{tikzpicture}
\\ (a) & (b)
\end{tabular}
  \end{figure}
  
  }

\begin{document}

\title{Consequences and applications of the completeness of Hardy's nonlocality
}

\author{Shane Mansfield}
\email[]{smansfie@staffmail.ed.ac.uk}
\affiliation{School of Informatics, University of Edinburgh, Informatics Forum, 10 Crichton Street, Edinburgh EH8 9AB, United Kingdom}
\date{\today}


\begin{abstract}
Logical nonlocality is completely characterised by Hardy's ``paradox'' in $(2,2,l)$ and $(2,k,2)$ scenarios. We consider a variety of consequences and applications of this fact. \begin{enumerate*}[label=(\roman*)]  \item Polynomial algorithms may be given for deciding logical nonlocality in these scenarios. \item Bell states are the only entangled two-qubit states which are not logically nonlocal under projective measurements. \item It is possible to witness Hardy nonlocality with certainty in a simple tripartite quantum system.
\item Non-commutativity of observables is necessary and sufficient for enabling logical nonlocality. \end{enumerate*}
\end{abstract}

\maketitle

\section{Introduction}

Since the fundamental insight of Bell~\cite{bell:64,bell:87}, it is known that quantum mechanics gives rise to stronger-than-classical, nonlocal correlations. Under seemingly natural assumptions of locality and realism, it can be shown that any empirical correlations should satisfy certain Bell inequalities, which can be violated quantum-mechanically, from which Bell's conclusion follows.

A more intuitive, logical approach to nonlocality proofs was pioneered by Heywood and Redhead \cite{redhead:83}, Greenberger, Horne, Shimony and Zeilinger~\cite{greenberger:89,greenberger:90} and Hardy~\cite{hardy:92,hardy:93} \footnote{See also Mermin's versions \cite{mermin:90a,mermin:90,mermin:94}. Several other logical nonlocality proofs have also appeared (e.g.~\cite{boschi:97,cereceda:04, ghosh:98}).}. This kind of nonlocality proof disregards the precise values of the probabilities for the various outcome events and only refers to events as being possible (with probability greater than zero) or impossible (having probability zero). This turns out to be sufficient for demonstrating nonlocality in quantum mechanics. We refer to these as \emph{logical nonlocality} proofs.

Probabilistic nonlocality, as witnessed by violations of Bell inequalities and logical nonlocality are the first two levels of a \emph{qualitative hierarchy} of nonlocality introduced in \cite{abramsky:11} \footnote{Note that the hierarchy emerges from the unified sheaf-theoretic framework for nonlocality and contextuality and so more generally applies to contextuality. Further refinements of the hierarchy may also be found in \cite{abramsky:11a} and \cite{abramsky:15contextuality}, and a related quantitative measure is also considered in \cite{abramsky:11,mansfield:13t,abramsky:16qpl,abramsky:16horse}.}, the highest level of which is \emph{strong} nonlocality, which arises when even at the level of possibilities the model cannot be factored into a local and a nonlocal part.

We work within a general framework, introduced in \cite{mansfield:11}, for logical nonlocality proofs in \emph{$(n,k,l)$ scenarios}---i.e., Bell-type scenarios in which $n$ is the number parties or sites, $k$ is the maximum number of measurement settings available at each site, and $l$ is the maximum number of potential outcomes for these measurements.
Our framework bears some similarity to the relational hidden variable framework of Abramsky~\cite{abramsky:10}, as well as a combinatorial framework due to Degorre and Mhalla \cite{degorre:10framework}, and while not as general could be considered a precursor to the unified sheaf-theoretic \cite{abramsky:11} and combinatorial \cite{acin:15} frameworks for nonlocality and contextuality \footnote{E.g.~The present framework corresponds to a purely \emph{possibilistic} version \cite{abramsky:16} of the sheaf-theoretic framework for Bell-type scenarios.}.
The advantage of the present framework is that it comes with a particular representation for $n=2$ and (as we will introduce in this article) $n=3$ scenarios, which can provide a powerful means of reasoning about \emph{empirical models}; i.e., probability or possibility tables for the various joint outcomes in a given scenario.

Hardy's logical nonlocality proof or ``paradox'' \cite{hardy:92,hardy:93} is often considered to be the simplest of all quantum mechanical nonlocality proofs. In \cite{mansfield:11}, the author and Fritz proved completeness results which establish that Hardy's paradox is a necessary and sufficient condition for logical nonlocality in all $(2,k,2)$ and $(2,2,l)$ scenarios (thereby subsuming all other logical nonlocality proofs or ``paradoxes'' in these scenarios). For the $(2,3,3)$ \cite{mansfield:11} and $(3,2,2)$ \cite{mansfield:13t} scenarios, it is known that this no longer holds.

In this article, we explore a variety of consequences and applications of the completeness of Hardy nonlocality. To begin with, we will see that in the relevant scenarios they lead to explicit algorithms for deciding logical nonlocality which are polynomial in $l$ and $k$.
They also lead to a constructive proof that the Popescu-Rohrlich box \cite{popescu:94} is the only \emph{strongly nonlocal} $(2,2,2)$ empirical model.

Next, we obtain a proof that Bell states are not logically nonlocal under projective measurements. Surprisingly, these are the \emph{only} entangled qubit states with this property: all other entangled two-qubit states have been shown to admit a Hardy paradox \cite{hardy:93}, and all entangled $n$-qubit states have also been shown to be logically nonlocal \cite{abramsky:16hardy}, both via appropriate choices of local projective measurements. In this sense, the Bell states are anomalous in the landscape of entangled states, in spite of the fact that they are among the most studied and utilised of these.

Much of the literature on Hardy's paradox is concerned with the \emph{paradoxical probability}; i.e., the probability of witnessing the particular outcome event from which the logical argument follows. This is often considered to be an indicator of the quality of Hardy nonlocality. For Hardy's family of quantum mechanical, nonlocal empirical models, the maximum paradoxical probability that can be achieved is $(5 \sqrt{5}-11)/2 \approx 0.09$. It has been shown, however, that it is possible to achieve a paradoxical probability of $0.125$ for a generalised version of Hardy's paradox in a tripartite quantum system \cite{ghosh:98}, and
it has also been shown that a ``ladder'' version of Hardy's paradox, which allows $k$ measurement settings to each party, can give rise to a higher paradoxical probability which approaches $0.5$ as $k \rightarrow \infty$.

More recently, Chen \textit{et al.}~found that another generalisation of Hardy's paradox can be witnessed with probability $\approx 0.4$ for certain high-dimensional bipartite quantum systems \cite{chen:13}. The measurement scenarios for these logical nonlocality proofs fall within the scope of the completeness results for Hardy nonlocality. We show explicitly that each Chen \textit{et al.}~paradox contains within it many different Hardy paradoxes. Moreover, we will see that their ``paradoxical probability'' might more accurately be described as the sum of the paradoxical probabilities for these Hardy paradoxes, all of which occur within the one model.

Using the completeness of Hardy nonlocality we will achieve a rather comprehensive improvement on these results, demonstrating by a much simpler argument that if such a summing of paradoxical probabilities is considered, it is possible to witness Hardy nonlocality \emph{with certainty} for a tripartite quantum system. Interestingly, the argument relies on the same state and measurements as the Greenberger-Horne-Zeilinger (GHZ) experiment \cite{greenberger:90}.
We also show that Hardy nonlocality can be achieved with certainty for a particular non-quantum, no-signalling $(2,2,2)$ model, which turns out to be the Popescu-Rohrlich no-signalling box \cite{popescu:94}.

Moreover, the notion of witnessing logical nonlocality with certainty corresponds precisely to the notion of strong nonlocality, the highest level in the qualitative hierarchy of nonlocality (the hierarchy also applies more generally to contextuality) introduced in \cite{abramsky:11}.

Finally, we employ the completeness results in order to prove that incompatibility of observables is necessary and sufficient for logical nonlocality, thus extending to the logical setting a result due to Wolf, Peres-Garcia and Fernandez \cite{wolf:09} which establishes that incompatibility is necessary and sufficient for (probabilistic) nonlocality.

\section{Logical Nonlocality and Hardy Paradoxes}\label{sec:background}

The possibility table (or \emph{possibilistic empirical model}) used for the original Hardy nonlocality proof, Table \ref{hptable} (a), concerns the $(2,2,2)$ scenario. Each of the two parties can make one of two measurements on their subsystem, giving rise to outcomes which we label here $\{ \uparrow, \downarrow\}$ for the first measurement and $\{R,G\}$ for the second. A $1$ in the table signifies that it is possible (with probability greater than zero) to obtain the corresponding joint outcome, and a $0$ signifies that it is not possible. The precise probabilities of obtaining the various joint outcomes are not required to prove the nonlocality of the model.
Any probabilistic empirical model can be transformed into a possibilistic empirical model of this kind in a canonical way via \emph{possibilistic collapse} \cite{mansfield:11,abramsky:11}: the process by which all non-zero probabilities are conflated to $1$, with zero probabilities mapping to $0$.

\begin{definition}
Any empirical model which is nonlocal at the level of its possibilistic table is said to be \emph{logically nonlocal}.
\end{definition}

\begin{proposition}[\cite{mansfield:11}]
A possibilistic empirical model is (logically) nonlocal if and only if it cannot be realised as a union of local deterministic models; or, equivalently, if there exists a $1$ in its possibility table which cannot be completed to a deterministic grid.
\end{proposition}

Local deterministic models are empirical models for which the outcome at each site is determined uniquely by the measurement at that site, and in the tabular representation take the form of \emph{deterministic grids}; e.g.~Table \ref{hptable} (b). Deterministic grids correspond to global sections of the event sheaf in the sheaf-theoretic approach \cite{mansfield:13t}, and indeed logical nonlocality is a special case of the general notion of \emph{contextuality} as considered in \cite{abramsky:11}, which is also proved there to be equivalent to the failure of a model to be realisable by a factorisable hidden variable model.

In the case of the Hardy paradox, it is clear that the $1$ in Table \ref{hptable} cannot be completed to a deterministic grid, regardless of the unspecified entries. However, depending on the scenario, this is just one way in which a model might exhibit nonlocality at the possibilistic level \cite{mansfield:11,mansfield:13t}.

\begin{definition}\label{def:hardynl}
Up to re-labelling of measurements and outcomes, any possibilistic $(2,2,2)$ empirical model containing the arrangement of $1$'s and $0$'s shown in Table \ref{hptable} (a) is said to \emph{contain a Hardy paradox} (i.e., it admits Hardy's logical nonlocality proof) and we say that the joint outcome $(\uparrow,\uparrow)$ \emph{witnesses Hardy nonlocality}.
\end{definition}

\begin{table}
\caption{\label{hptable} (a) A possibilistic empirical model containing a Hardy paradox. This is a possibility table in which $1$ denotes ``possible'' and $0$ denotes ``impossible''. The blank entries are not relevant and may each take either of the values. (b) A ``deterministic grid'' or local deterministic model.}
\begin{center}
\begin{tabular}{cc}
\begin{tabular}{cc}
~ & Bob \\
Alice & \begin{tabular}{c|>{\centering\arraybackslash}p{.4cm}>{\centering\arraybackslash}p{.4cm}|>{\centering\arraybackslash}p{.4cm}>{\centering\arraybackslash}p{.4cm}|}
~ & $\uparrow$ & $\downarrow$ & $R$ & $G$ \\ \hline
$\uparrow$ & 1 & ~ & ~ & 0 \\
$\downarrow$ & ~ & ~ & ~ & ~ \\ \hline
$R$ & ~ & ~ & 0 & ~ \\
$G$ & 0 & ~ & ~ & ~ \\ \hline
\end{tabular}
\end{tabular}
&
\begin{tabular}{cc}
~ & Bob \\
Alice & \begin{tabular}{c|>{\centering\arraybackslash}p{.4cm}>{\centering\arraybackslash}p{.4cm}|>{\centering\arraybackslash}p{.4cm}>{\centering\arraybackslash}p{.4cm}|}
~ & $\uparrow$ & $\downarrow$ & $R$ & $G$ \\ \hline
$\uparrow$ & \textcolor{blue}{1} & \textcolor{gray}{0} & \textcolor{blue}{1} & \textcolor{gray}{0} \\
$\downarrow$ & \textcolor{gray}{0} & \textcolor{gray}{0} & \textcolor{gray}{0} & \textcolor{gray}{0} \\ \hline
$R$ & \textcolor{blue}{1} & \textcolor{gray}{0} & \textcolor{blue}{1} & \textcolor{gray}{0} \\
$G$ & \textcolor{gray}{0} & \textcolor{gray}{0} & \textcolor{gray}{0} & \textcolor{gray}{0} \\ \hline
\end{tabular}
\end{tabular}
\\
~ & ~ \\
(a) & (b)
\end{tabular}
\end{center}
\end{table}

Definition \ref{def:hardynl} defines Hardy nonlocality for $(2,2,2)$ scenarios. It is also possible to extend the definition to $(2,2,l)$ scenarios simply by course-graining outcomes; see Table \ref{genpar}. Furthermore, one may define Hardy nonlocality in $(2,k,l)$ models as arising whenever some $2 \times 2$ subtable (i.e., restricting attention to any two of the $k$ measurements at each site) contains a Hardy paradox; see Table \ref{2k2par}.

\begin{definition}
Any possibilistic $(2,k,l)$ empirical model containing a $2 \times 2$ subtable which is isomorphic (up to coarse-graining of outcomes and re-labelling of measurements and outcomes) to Table~\ref{genpar} is said to contain a \emph{(coarse-grained) Hardy paradox}.
\end{definition}

\begin{table}
\caption{\label{genpar} A $(2,2,l)$ scenario with a coarse-grained Hardy paradox.}
\begin{center}
\begin{tabular}{c|c>{\centering\arraybackslash}p{2.2cm}c|cc|}
~ & $o_1'$ & $\cdots$ & $o_l'$ & $o_1 \cdots o_{m_2}$ & $o_{m_2+1} \cdots o_l$ \\ \hline
$o_1'$ & $1$ & ~ & ~ & ~ & $0 \quad \cdots \quad 0$ \\
$\vdots$ & ~ & ~ & ~ & ~  & ~ \\
$o_l'$ & ~ & ~ & ~ & ~  & ~ \\ \hline
$\begin{array}{c}
o_1 \\ \vdots \\ o_{m_1}
\end{array}$
& ~ & ~ & ~ &
$\begin{array}{ccc}
0 & \cdots & 0 \\ \vdots & \ddots & \vdots \\ 0 & \cdots & 0
\end{array}$
& ~ \\
$\begin{array}{c}
o_{m_1+1} \\ \vdots \\ o_l
\end{array}$
& 
$\begin{array}{c}
0 \\ \vdots \\ 0
\end{array}$
& ~ & ~ & ~ & ~ \\ \hline
\end{tabular}
\end{center}
\end{table}

\begin{table}
\caption{\label{2k2par} A $(2,k,2)$ scenario containing a Hardy paradox.}
\begin{center}
\begin{tabular}{c|c|c|c|}
~ & ~ & ~ & ~ \\ \hline
~ & \begin{tabular}{cc}
1 & ~ \\ ~ & ~
\end{tabular}
& $\cdots$ &
\begin{tabular}{cc}
~ & 0 \\ ~ & ~
\end{tabular}
\\ \hline
~ & \vdots & ~
\begin{tabular}{ccc}
~ & ~ & ~ \\ ~ & $\ddots$ & ~ \\ ~ & ~ & ~
\end{tabular}
& $\vdots$ \\ \hline
~ &
\begin{tabular}{cc}
~ & ~ \\ 0 & ~
\end{tabular}
& $\cdots$ &
\begin{tabular}{cc}
0 & ~ \\ ~ & ~
\end{tabular} 
\\ \hline
\end{tabular}
\end{center}
\end{table}

Wang and Markham have described a generalisation of Hardy's logical nonlocality proof to $(n,2,2)$ scenarios, which they have used to demonstrate that all symmetric $n$-partite qubit states for $n>2$ admit logical nonlocality proofs~\cite{wang:12}. This kind of generalisation has been described elsewhere by Ghosh, Kar and Sarkar~\cite{ghosh:98}, and is also considered in \cite{cereceda:04} and \cite{choudhary:08}.

We write $p(o \mid m) = 1$ if it is possible with probability greater than zero to obtain joint outcome $o$ when joint measurement $m$ is made, and $p(o \mid m) = 0$ otherwise. Here, $0$ and $1$ play the role of Boolean truth values. For $(n,2,2)$ scenarios we also let measurements and outcomes both be labelled by $\{0,1\}$ at each site, though note that these $0$'s and $1$'s simply play the role of labels.

\begin{definition}
For any $(n,2,2)$ scenario, an \emph{$n$-partite Hardy paradox} occurs if (up to re-labelling of measurements and outcomes) the following possibilistic conditions are satisfied.
\begin{itemize}
\item
$p(\,0, \dots , 0\,\mid\,0,\dots,0\,)=1$
\item
$p\left(\,\pi(1,0,\dots,0)\,\mid\,\pi(1,0,\dots,0)\,\right)=0$, \\ for all permutations $\pi \in S^n$
\item
$p(\,0,\dots,0\,\mid\,1,\dots,1\,)=0$
\end{itemize}
\end{definition}

The $n=3$ generalisation of the Hardy paradox can also be represented in a three-dimensional version of the tabular representation; see Fig.~\ref{fig:3dhardy}\index{empirical model!3d representation}. The advantage of the representation, as we will see, is that it provides a powerful visual means of analysing models. The axes again correspond to different sites, the eight medium-sized cubes to joint measurement settings, and the smallest sub-cubes to joint outcomes, similarly to the $n=2$ case. The properties of the tabular representation generalise in the obvious way to the third dimension.

\drawNKLCube{
	\pointAt{0}{0}{3}{blue!50}
	
	\pointAt{0}{0}{0}{red!50}
	\pointAt{3}{0}{3}{red!50}
	\pointAt{0}{3}{3}{red!50}
	\pointAt{2}{2}{1}{red!50}
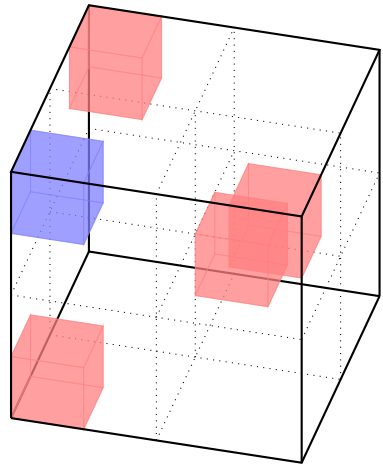
}{\label{fig:3dhardy} The $n=3$ Hardy paradox. The blue entry (upper-left sub-cube) corresponds to a possible joint outcome, and the red entries (other sub-cubes) to impossible ones. The blank entries are not relevant, and may take either of the values.}

For example, in \cite{mansfield:11} it was shown that for $(2,k,l)$ scenarios a possibilistic empirical model is local if and only if every $1$ in its table can be completed to a deterministic grid.
This characterisation generalises in the obvious way to the three-dimensional representation for $n=3$ models, so that we can similarly see by inspection that the blue entry in Fig.~\ref{fig:3dhardy} cannot be completed to a (three-dimensional) deterministic grid, just as the $1$ in Table \ref{hptable} cannot be completed to a deterministic grid, and therefore any $(3,2,2)$ model containing this arrangement of $1$'s and $0$'s, or red and blue boxes, is logically nonlocal.

It is known that Hardy nonlocality completely characterises logical nonlocality in a variety of scenarios. The following theorem combines the completeness results of \cite{mansfield:11}.

\begin{theorem}[Mansfield and Fritz \cite{mansfield:11}]\label{equivprop}
For any $(2,k,2)$ or $(2,2,l)$ scenario, an empirical model is logically nonlocal if and only if it contains a (coarse-grained) Hardy paradox.
\end{theorem}

We also rephrase the definition of strong nonlocality as introduced in \cite{abramsky:11} within the present framework.

\begin{definition}
An empirical model is \emph{strongly nonlocal} if and only if \emph{no} $1$ in its possibility table can be completed to a deterministic grid.
\end{definition}

Hardy and logical nonlocality are situated within the qualitative hierarchy of increasing strengths of nonlocality as follows:
\begin{equation}\label{eq:hierarchy}
\text{probabilistic} \, < \, \text{Hardy} \, < \, \text{logical} \, < \, \text{strong},
\end{equation}
where membership of any of these classes implies membership of all lower classes. At the lowest level, a model is probabilistically nonlocal if and only if it violates some Bell inequality. The hierarchy is in general strict: for each class, empirical models can be found which do not belong to any higher class. For measurement scenarios in which Theorem~\ref{equivprop} applies, however, the Hardy and logical classes coincide.

\section{Complexity of Logical Nonlocality}\label{sec:complexity}

Theorem \ref{equivprop} is relevant to the computational complexity of deciding logical nonlocality in $(2,2,l)$ and $(2,k,2)$ scenarios, where it is equivalent to deciding whether a Hardy paradox occurs. The fact was mentioned in \cite{mansfield:11}; here we find explicit polynomial algorithms.

\begin{proposition}\label{prop:complexity}
Polynomial algorithms can be given for deciding nonlocality in $(2,2,l)$ and $(2,k,2)$ models.
\end{proposition}

\begin{proof}
For $(2,k,2)$ scenarios, deciding whether a model in the tabular form is local or nonlocal simply amounts to checking all $2\times 2$ sub-tables for such a Hardy paradox, which gives an algorithm that is polynomial in the size of the input table: we check for the $64$ possible Hardy configurations in each of ${k \choose 2}^2$ sub-tables, which is $O(k^4)$. For $(2,2,l)$ scenarios, one has to check each $1$ in the table to see whether it can be completed to a deterministic grid. There are $4l^2$ entries in the table, and each check is $O(l^2)$, so again we have an algorithm that is polynomial in the size of the input.
\end{proof}

It was conjectured in \cite{mansfield:11} that decidability of logical nonlocality with $k$ as the free input is NP-hard when $n>2,l\geq2$ or $n\geq2,l>2$, as is known to be the case for probabilistic models \cite{pitowsky:91}. The problem was shown to be NP by Abramsky in \cite{abramsky:10}, and the it has since been proved to be NP-complete by Abramsky, Gottlob and Kolaitis \cite{abramsky:13b}. This gives strong reason to suspect that it is not possible to obtain a classification of conditions that are necessary and sufficient for logical nonlocality in full generality.

\section{Strong Nonlocality and the PR Box}\label{sec:pr}

Recall from \eqref{eq:hierarchy} that strong nonlocality is strictly stronger than logical nonlocality.
Theorem \ref{equivprop} can be used to give a constructive proof of a result originally proved by case-analysis by Lal \cite{abramsky:11,lal:13} that the only strongly nonlocal $(2,2,2)$ models are the Popescu-Rohrlich no-signalling boxes \cite{popescu:94}, whose probability table up to re-labelling is given in Table~\ref{tab:pr}.

\begin{proposition}\label{prop:sc}\index{nonlocality!strong nonlocality}
The only strongly nonlocal no-signalling $(2,2,2)$ models are the PR boxes.
\end{proposition}

Before we prove this proposition, recall that by Theorem \ref{equivprop} strong nonlocality is equivalent to the property that every $1$ in its possibility table witnesses a Hardy paradox. We will simply use this property together with the requirement that the model satisfies no-signalling to derive our result. An illustration of no-signalling in the possibilistic sense is the following. We see from Table~\ref{hptable}~(a) that if Alice and Bob each make their $\{\uparrow,\downarrow\}$ measurement then it is possible for Alice to obtain the outcome $\uparrow$. Now in order to make sure that Bob cannot instantaneously signal to Alice who is assumed to be spacelike separated from him it must be the case that it would also be possible for Alice to obtain the outcome $\uparrow$ had Bob made his $\{R,G\}$ measurement. We can therefore deduce that since the event $(\uparrow_A,G_B)$ is not possible the event $(\uparrow_A,R_B)$ must be possible. More generally in the tabular representation, no-signalling translates to the condition that whenever a $1$ occurs in a table, the outcome row and column the event belongs to must each contain at least one $1$ per measurement setting, for otherwise the possibility of witnessing a particular outcome for one party could depend on the measurement choice of the other (see \cite{mansfield:11} for a more detailed discussion).

\begin{proof}
For any choice of measurements there must be \emph{some} possible outcome. This possible assignment is represented by a $1$ in the table, and it must witness a Hardy paradox. After re-labelling as necessary, we can represent the model as in Table \ref{hptable} (a). For this to be a no-signalling model, it is necessary to fill in $1$'s as in Table~\ref{tab:sc}~(a). Using the fact that the $1$'s in the lower-right box must also witness Hardy paradoxes, we must fill in $0$'s as in Table \ref{tab:sc} (b). By no-signalling, the remaining unspecified entry in the upper-left box must be a $1$, and by the fact that it must witness a Hardy paradox, the remaining entry in the lower-right box must be a $0$. We thus arrive at Table~\ref{tab:sc}~(c), and the unique no-signalling probabilistic empirical model whose possibility table has this form is the PR box.
\end{proof}

\begin{table}
\caption{\label{tab:sc} Stages in the proof of proposition~\ref{prop:sc}.}
\begin{center}
\begin{tabular}{ccc}
\begin{tabular}{c|cc|cc|}
~ & ~ & ~ & ~ & ~ \\ \hline
~ & 1 & ~ & \textcolor{blue}{1} & 0 \\
~ & ~ & ~ & ~ & \textcolor{blue}{1} \\ \hline
~ & \textcolor{blue}{1} & ~ & 0 & \textcolor{blue}{1} \\
~ & 0 & \textcolor{blue}{1} & \textcolor{blue}{1} & ~ \\ \hline
\end{tabular}
\quad & \quad
\begin{tabular}{c|cc|cc|}
~ & ~ & ~ & ~ & ~ \\ \hline
~ & 1 & \textcolor{blue}{0} & 1 & 0 \\
~ & \textcolor{blue}{0} & ~ & \textcolor{blue}{0} & 1 \\ \hline
~ & 1 & \textcolor{blue}{0} & 0 & 1 \\
~ & 0 & 1 & 1 & ~ \\ \hline
\end{tabular}
\quad & \quad
\begin{tabular}{c|cc|cc|}
~ & ~ & ~ & ~ & ~ \\ \hline
~ & 1 & 0 & 1 & 0 \\
~ & 0 & \textcolor{blue}{1} & \textcolor{blue}{0} & 1 \\ \hline
~ & 1 & \textcolor{blue}{0} & 0 & 1 \\
~ & 0 & 1 & 1 & \textcolor{blue}{0} \\ \hline
\end{tabular} \\
~ & ~ & ~ \\
(a) \quad & \quad (b) \quad & \quad (c)
\end{tabular}
\end{center}
\end{table}

\section{Bell State Anomaly}

It is known how to prescribe projective measurements for almost all entangled two-qubit states such that the resulting empirical model will contain a Hardy paradox \cite{hardy:93}, the exception being the maximally entangled states; i.e., the familiar Bell states. This naturally raises the question of whether there exist any projective measurements that can be chosen for the maximally entangled states such that the resulting empirical model contains a Hardy paradox. Indeed, in light of Theorem \ref{equivprop} we know that this is equivalent to asking whether the maximally entangled states are logically nonlocal under projective measurements. Some previous failed attempts at finding a logical nonlocality proof for the Bell states are described in \cite{cabello:00}.

We answer this question in the negative, and show that no projective measurements can be chosen that lead to a Hardy paradox (and thus logical nonlocality). A result showing that if the same pair of local measurements are available at each qubit then it is impossible to realise a Hardy paradox was proved independently by Abramsky and Constantin \cite{abramsky:13}, but the theorem we are about to present holds for any number of measurements per qubit, and without the restriction that the same set of local measurements be available at each qubit.

In fact, Bell states are the only entangled $n$-qubit states, for any $n$, which are not logically nonlocal under projective measurements, since for $n>2$ it is known that projective measurements can be found for all $n$-qubit entangled states which give rise to logical nonlocality \cite{abramsky:16hardy}.
In this sense, despite being among the most studied and utilised states in the fields of quantum information and computation, the Bell states are actually anomalous in the landscape of entangled states.

\begin{theorem}\label{thm:ba}
Bell states are not logically nonlocal under projective measurements.
\end{theorem}

\begin{proof}
We prove the statement for the Bell state
\[
\ket{\phi^+}  = \frac{1}{\sqrt{2}} \left( \ket{00} + \ket{11} \right).
\]
Since all other maximally entangled states are equivalent to this one up to local unitaries, which can easily be incorporated into the local measurements, the proof will extend to all maximally entangled states.

Any quantum mechanical empirical model obtained by making local projective measurements on $\ket{\phi^+}$ will necessarily give rise to a $(2,k,2)$ model. By Theorem \ref{equivprop} we know that Hardy's paradox completely characterises logical nonlocality for such scenarios, and that logical nonlocality would therefore imply the occurrence of a Hardy paradox in some $(2,2,2)$ sub-model. It therefore suffices to show that for any observables $A_1$ and $A_2$ for the first qubit and $B_3$ and $B_4$ for the second qubit the resulting model does not contain a Hardy paradox.

The $+1$ and $-1$ eigenvectors for these measurements will be given by
\begin{align*}
\ket{0_i} &= \cos{\frac{\theta_i}{2}} \ket{0} + e^{i \phi_i} \sin{\frac{\theta_i}{2}} \ket{1} \\
\ket{1_i} &= \sin{\frac{\theta_i}{2}} \ket{0} + e^{- i \phi_i} \cos{\frac{\theta_i}{2}} \ket{1}
\end{align*}
where $\{ ( \theta_i, \phi_i) \}_{i \in \{1, 2, 3, 4\} }$ label the coordinates of the $+1$ eigenvector of the respective measurements on the Bloch sphere. The amplitudes for the outcomes of the various joint measurements are calculated to be:
\begin{align*}
\braket{0_j0_k}{\phi^+} &= \frac{1}{\sqrt{2}} \left( \cos{ \frac{\theta_j}{2}} \cos{\frac{\theta_k}{2}} + e^{-i \left( \phi_j + \phi_k \right)} \sin{ \frac{ \theta_j}{2}} \sin{ \frac{\theta_k}{2}} \right) \\
\braket{0_j1_k}{\phi^+} &= \frac{1}{\sqrt{2}} \left( \cos{ \frac{\theta_j}{2}} \sin{\frac{\theta_k}{2}} + e^{-i \left( \phi_j - \phi_k \right)} \sin{ \frac{ \theta_j}{2}} \cos{ \frac{\theta_k}{2}} \right) \\
\braket{1_j0_k}{\phi^+} &= \frac{1}{\sqrt{2}} \left( \sin{ \frac{\theta_j}{2}} \cos{\frac{\theta_k}{2}} + e^{i \left( \phi_j - \phi_k \right)} \sin{ \frac{ \theta_j}{2}} \cos{ \frac{\theta_k}{2}} \right) \\
\braket{1_j1_k}{\phi^+} &= \frac{1}{\sqrt{2}} \left( \sin{ \frac{\theta_j}{2}} \sin{\frac{\theta_k}{2}} + e^{i \left( \phi_j + \phi_k \right)} \cos{ \frac{ \theta_j}{2}} \cos{ \frac{\theta_k}{2}} \right)
\end{align*}
where $j \in \{1,2\}$ and $k \in \{3,4\}$. We see that $\braket{0_j0_k}{\phi^+}=e^{-i\left(\phi_j+\phi_k\right)} \braket{1_j1_k}{\phi^+}$ and $\braket{0_j1_k}{\phi^+} = \braket{1_j0_k}{\phi^+}$ for each choice of measurements. Thus the symmetry of the underlying state manifests itself as a symmetry in the probabilities of the joint outcomes for each choice of measurements:
\begin{align}
p(01\mid AB) &= p(10\mid AB) \label{eq:sym1} , \\
p(00\mid AB) &= p(11\mid AB) \label{eq:sym2}.
\end{align}

Note that the PR box (Table~\ref{tab:pr}), which we know from Proposition \ref{prop:sc} to be the only strongly nonlocal $(2,2,2)$ model (up to re-labellings), satisfies these symmetries. However, it is also known that the PR box is not quantum-realisable \cite{popescu:94,tsirelson:80}, so while it satisfies the symmetries it nevertheless cannot be realised by measurements on $\ket{\phi^+}$.

\begin{table}
\caption{\label{tab:ba} Stages in the proof of Theorem \ref{thm:ba}.}
\begin{center}
\begin{tabular}{cccc}
\begin{tabular}{c|cc|cc|}
~ & ~ & ~ & ~ & ~ \\ \hline
~ & 1 & ~ & 1 & ~ \\
~ & ~ & ~ & ~ & ~ \\ \hline
~ & 1 & ~ & 1 & ~ \\
~ & ~ & ~ & ~ & ~ \\ \hline
\end{tabular}
\quad & \quad
\begin{tabular}{c|cc|cc|}
~ & ~ & ~ & ~ & ~ \\ \hline
~ & 1 & ~ & 1 & ~ \\
~ & ~ & \textcolor{blue}{1} & ~ & \textcolor{blue}{1} \\ \hline
~ & 1 & ~ & 1 & ~ \\
~ & ~ & \textcolor{blue}{1} & ~ & \textcolor{blue}{1} \\ \hline
\end{tabular}
\quad & \quad
\begin{tabular}{c|cc|cc|}
~ & ~ & ~ & ~ & ~ \\ \hline
~ & 1 & ~ & 1 & \textcolor{blue}{0} \\
~ & \textcolor{blue}{0} & 1 & ~ & 1 \\ \hline
~ & 1 & ~ & 1 & ~ \\
~ & \textcolor{blue}{1} & 1 & \textcolor{blue}{0} & 1 \\ \hline
\end{tabular}
\quad & \quad
\begin{tabular}{c|cc|cc|}
~ &\multicolumn{2}{c|}{$B_3$}  & \multicolumn{2}{c|}{$B_4$} \\ \hline
\multirow{2}{*}{$A_1$} & 1 & 0 & 1 & 0 \\
& 0 & 1 & 0 & 1 \\ \hline
\multirow{2}{*}{$A_2$} & 1 & 1 & 1 & 0 \\
~ & 1 & 1 & 0 & 1 \\ \hline
\end{tabular} \\
~ & ~ & ~ & ~ \\
(a) \quad & \quad (b) \quad & \quad (c) \quad & \quad (d)
\end{tabular}
\end{center}
\end{table}

Next, we show that there is a unique possibilistic $(2,2,2)$ model (up to re-labelling) which satisfies the symmetries (\ref{eq:sym1}) and (\ref{eq:sym2}) and is logically but not strongly nonlocal.
If a model is not strongly nonlocal then there exists at least one global assignment compatible with the model, or in tabular form at least one deterministic grid. Up to re-labelling this is represented in Table \ref{tab:ba} (a). By the symmetry (\ref{eq:sym2}) there must exist a second global assignment, as in Table \ref{tab:ba} (b). It is clear from the configuration of the table that none of the entries that have already been specified can witness a Hardy paradox. If the model is logically nonlocal, therefore, at least one of the unspecified entries in Table \ref{tab:ba} (b) must witness a Hardy paradox. Up to re-labelling, this can be represented as in Table~\ref{tab:ba}~(c). By the symmetry (\ref{eq:sym1}) the table must be completed to Table \ref{tab:ba} (d). This (up to re-labelling) is the only possibilistic empirical model that respects the symmetries and is logically nonlocal without being strongly nonlocal. The question now is whether it can be realised by measurements on $\ket{\phi^+}$.

Consider the measurement statistics for the joint measurement $A_1B_3$ required by Table~\ref{tab:ba}~(d). If these are to arise from quantum observables $A_1$ and $B_3$, then $\braket{\phi^+}{0_10_3} = \braket{\phi^+}{1_11_3} = \frac{1}{\sqrt{2}}$ and $\braket{\phi^+}{0_11_3} = \braket{\phi^+}{1_10_3} = 0$. So, either $\ket{0_1} = \ket{0_3} = \ket{0}$ and $\ket{1_1} = \ket{1_3} = \ket{1}$ up to an overall sign or vice versa. The eigenvectors of both observables are $\{\ket{0},\ket{1}\}$, so they must simply be Pauli $X$ operators (up to a common sign, which would allow for re-labelling the outcomes): 
\begin{equation}\label{eq:baconst1}
A_1 = B_3 = \pm X. 
\end{equation}
A similar argument applies for the joint measurements $A_1B_4$ and $A_2B_4$, showing that
\begin{align}
A_1 &= B_4 = \pm X, \\
A_2 &= B_4 = \pm X. \label{eq:baconst3}
\end{align}
Eqs.~(\ref{eq:baconst1})--(\ref{eq:baconst3}) imply that
\begin{equation*}
A_1= A_2 =B_3=B_4 = \pm X;
\end{equation*}
but therefore the measurement statistics for $A_2B_3$ must be the same as for each of the other joint measurements, and Table~\ref{tab:ba}~(d) is not realised.
This completes the proof that no quantum mechanical logically nonlocal empirical model can be obtained by considering (any number of) local projective measurements on the Bell state.
\end{proof}

Symmetry is important here: the symmetry of the underlying state manifests itself as a symmetry of the probabilities of outcomes for each joint measurement, \eqref{eq:sym1} and \eqref{eq:sym2}. By Theorem~\ref{equivprop}, logical nonlocality also requires a particular relationship between certain probabilities in each of these distributions (a Hardy paradox). However, quantum mechanically, there cannot exist local projective measurements that realise these correlations and respect the symmetries at the same time. On the other hand, there exists a whole family of no-signalling empirical models which are logically nonlocal and respect the symmetries. These are the no-signalling models with support as in Table \ref{tab:ba} (d), along with the PR box.

These models have some interesting properties in their own right \cite{mansfield:13t}: despite not being realisable quantum mechanically, they may lie within the Tsirelson bound, coming arbitrarily close to the local polytope. They can be seen, however, to violate information causality, which has been proposed as a physical principle that might characterise quantum correlations \cite{pawlowski:09} or ``almost quantum'' correlations \cite{navascues:15almost}, by means of the same protocol described in \cite{pawlowski:09}. In fact, similar families of models to this one have already been considered in this context in \cite{allcock:09}.

We also note that Fritz \cite{fritz:11} has considered quantum analogues of Hardy's paradox\index{Hardy's paradox!quantum analogues}. These are not realisable quantum mechanically, but can arise in more general no-signalling empirical models. An interesting point is that Table~\ref{tab:ba}~(d) contains two such paradoxes, and so the fact that any model with this support is not quantum-realisable also follows  more directly from this observation.

\section{Hardy Subsumes Other Paradoxes}\label{sec:chen}

An immediate consequence of Theorem \ref{equivprop} is that in the relevant scenarios Hardy's paradox subsumes all other paradoxes, in the sense that any model which can be demonstrated to be logically nonlocal necessarily contains a Hardy paradox. For instance, the ladder paradox \cite{boschi:97} has been proposed as a generalisation of the original Hardy paradox and was used for experimental tests of quantum nonlocality \cite{barbieri:05}. Up to symmetries, there is one ladder paradox for any number of settings $k$; i.e., for each $(2,k,l)$ scenario. It was observed in \cite{mansfield:11} that, by Theorem \ref{equivprop}, any ladder paradox necessarily contains a Hardy paradox, and, moreover, explicitly demonstrated how this comes about.

Here we consider a more recent proposal by Chen \textit{et al.}~\cite{chen:13} for an alternative generalisation of Hardy's paradox for high-dimensional (qudit) systems (see Table \ref{tab:chen}); this will also be relevant to the discussion in Sec.~\ref{sec:hpcert}. In the present terminology, the argument applies to $(2,2,l)$ Bell scenarios.

\begin{table}
\caption{\label{tab:chen} The Chen \textit{et al.}~paradox occurs when at least one of the starred entries is non-zero. The relevant outcomes for each joint measurement are either those above or those below the diagonal.}
\begin{center}
\begin{tabular}{c|cc>{\centering\arraybackslash}p{0.9cm}c|cccc|}
~ & ~ & ~ & ~ & ~ & ~ & ~ & ~ & ~ \\ \hline
~ & ~ & * & $\cdots$ & * & ~ & 0 & $\cdots$ & 0 \\
~ & ~ & ~ & $\ddots$ & $\vdots$ & ~ & ~ & $\ddots$ & $\vdots$ \\
~ & ~ & ~ & ~ & * & ~ & ~ & ~ & 0 \\
~ & ~ & ~ & ~ & ~ & ~ & ~ & ~ & ~ \\ \hline
~ & ~ & 0 & $\cdots$ & $0$ & ~ & ~ & ~ & ~ \\
~ & ~ & ~ & $\ddots$ & $\vdots$ &$0$ & ~ & ~ & ~ \\
~ & ~ & ~ & ~ & 0 & $\vdots$ & $\ddots$ & ~ & ~ \\
~ & ~ & ~ & ~ & ~ & 0 & $\cdots$ & $0$ & ~ \\ \hline
\end{tabular}
\end{center}
\end{table}

\begin{proposition}\label{prop:chen}
The occurrence of a Chen \textit{et al.}~paradox (Table \ref{tab:chen}) implies the occurrence of a Hardy paradox.
\end{proposition}

\begin{proof}
This follows directly from Theorem \ref{equivprop}, but one can also prove the proposition more directly. Suppose one of the starred entries corresponding to outcomes $(o'_i,o_j)$ of Table \ref{tab:chen} is non-zero. We write $p(i,j)>0$ for short. Then we can see from the table that for the joint measurement represented by the upper-right box, we must have $p(r,j)=0$ for all $r>(l-j)$. Similarly, for the measurement represented by the lower-left box, $p(i,s)=0$ for all $s>(l-i)$. In the lower-right box, we have $p(r,s) = 0$ when $r\leq (l-j)$ and $s\leq (l-i)$. This describes a $(2,2,l)$ Hardy paradox.
\end{proof}

The proof shows that every non-zero starred entry in Table~\ref{tab:chen} witnesses a (coarse-grained) Hardy paradox.

\section{Hardy Nonlocality with Certainty}\label{sec:hpcert}

While Hardy's paradox is considered to be an ``almost probability free'' nonlocality proof, much of the literature on Hardy's paradox has been concerned with the value of the \emph{paradoxical probability} (e.g.~\cite{boschi:97,ghosh:98,choudhary:08,chen:13}); i.e., the probability of obtaining the particular outcome that witnesses a Hardy paradox (Definition \ref{def:hardynl}). This is motivated as being especially relevant for experimental tests. In this section, we will show how Hardy nonlocality can be demonstrated in such a way that even this probability becomes irrelevant.

We note that similar argument was put forward by Cabello \cite{cabello:98}, but stress that the results contained in this section has the advantage of being far simpler, both in terms of the argument and of the empirical models in question.

As previously mentioned, Hardy \cite{hardy:93} prescribed measurements for all entangled two-qubit states (excluding the maximally entangled ones) such that the resulting empirical model contains a Hardy paradox. The maximum paradoxical probability over this family of quantum-realisable empirical models is
\begin{equation}\label{hardybound}\index{Hardy's paradox!bound}
p_{\max} = \frac{5\sqrt{5}-11}{2} \approx 0.09\,.
\end{equation}
A model has also been found for which the tripartite Hardy paradox can be witnessed with probability $0.125$ \cite{ghosh:98}, and in \cite{mermin1995best,cereceda2000quantum,choudhary:08} it is demonstrated that for a generalised no-signalling theory it is possible to witness a $(2,2,2)$ Hardy paradox with probability $0.5$. It was shown that the ladder generalisation of Hardy's paradox could achieve a paradoxical probability approaching $0.5$ for $(2,k,2)$ scenarios, as $k \rightarrow \infty$. For the $(2,2,l)$ scenario, Chen \textit{et al.}~\cite{chen:13} (cf.~Sec.~\ref{sec:chen}) have claimed that it is possible to achieve a paradoxical probability of $\approx 0.4$ in the large $d$ limit for two qu$d$it systems with the paradox presented in Table \ref{tab:chen}. From our Proposition \ref{prop:chen}, it is clear that strictly speaking this comes about by summing the probabilities of witnessing a number of different (coarse-grained) Hardy paradoxes; $(l-1)^2/2$ of these to be precise.

In this section, we use completeness of Hardy nonlocality to achieve a comprehensive improvement on these results, demonstrating by simple arguments that, by considering such a summation of different paradoxical probabilities, Hardy nonlocality can in fact be witnessed with certainty in a tripartite quantum system. This turns out to be demonstrable with the familiar GHZ-Mermin model \cite{greenberger:89,greenberger:90,mermin:90a,mermin:90}. We will first show that the property also holds for a particular no-signalling but non-quantum $(2,2,2)$ empirical model, which turns out to be the PR box.

\begin{proposition}
The PR box witnesses Hardy nonlocality with certainty.
\end{proposition}

\begin{proof}
The probabilistic version of the PR box is given in Table \ref{tab:pr}. We have already observed in the proof of Proposition \ref{prop:sc} that every joint outcome that has non-zero probability witnesses a Hardy paradox. Therefore, each non-zero entry in the table represents a joint outcome that witnesses Hardy nonlocality with paradoxical probability $0.5$, and so it is clear that for each joint measurement the probability of obtaining an outcome that witnesses a Hardy paradox is $1$.
\end{proof}

\begin{table}
\caption{\label{tab:pr} The PR box.}
\begin{center}
\begin{tabular}{c|cc|cc|}
~ & ~ & ~ & ~ & ~ \\ \hline
~ & \sfrac{1}{2} & 0 & \sfrac{1}{2} & 0 \\
~ & 0 & \sfrac{1}{2} & 0 & \sfrac{1}{2} \\ \hline
~ & \sfrac{1}{2} & 0 & 0 & \sfrac{1}{2} \\
~ & 0 & \sfrac{1}{2} & \sfrac{1}{2} & 0 \\ \hline
\end{tabular}
\end{center}
\end{table}

It is clear that the PR box achieves the upper bound on paradoxical probabilities for individual Hardy paradoxes in no-signalling models of $0.5$
\cite{mermin1995best,cereceda2000quantum} and provides a more concise example a model saturating the bound than that constructed in \cite{choudhary:08}. Perhaps more importantly, however, we see that since \emph{every} joint outcome witnesses a Hardy paradox in the present example, the arguably more relevant parameter, the probability of witnessing Hardy nonlocality, is actually $1$ for any choice of measurements.

Nevertheless, it is not possible to use this method of summing paradoxical probabilities to witness Hardy nonlocality with higher probability than \eqref{hardybound} for any $(2,k,2)$ empirical model which can be realised by projective measurements on a Bell state.

\begin{proposition}
For any $(2,k,2)$ empirical model arising from projective measurements on a Bell state, the probability of witnessing Hardy nonlocality cannot be improved by summing the paradoxical probabilities for different paradoxes occurring within the same model.
\end{proposition}

\begin{proof}
First, we note that it suffices to prove the proposition for $(2,2,2)$ models, since a $(2,k,2)$ model contains a Hardy paradox if and only if some $(2,2,2)$ sub-model contains a Hardy paradox. In order to obtain an improvement in the probability of witnessing Hardy nonlocality it would have to be the case that, for some joint measurement, more than one Hardy paradox could be witnessed. Working in the present framework, it is clear that any such empirical model is either the PR box or belongs to the family of models with support given by Table~\ref{tab:ba}~(d), up to re-labelling of measurements and outcomes, as discussed in the proof of Theorem~\ref{thm:ba}. Indeed, in this family, for the joint measurement $A_2B_3$, the probability of witnessing Hardy nonlocality is $1$. However, it was also shown in the proof of Theorem~\ref{thm:ba} that no model in the family is quantum-realisable.
\end{proof}


We now consider the $(3,2,2)$ empirical model used in the GHZ-Mermin logical nonlocality proof \cite{greenberger:90,mermin:90a}. It should be noted that that the original nonlocality argument based on this empirical model was not of the tripartite Hardy form mentioned in Sec.~\ref{sec:background}. Here, we need only consider a subset of the measurement contexts, shown in Table \ref{tab:ghz} in more orthodox notation, and three-dimensional representation in Fig.~\ref{fig:ghz3d} (a).
\begin{table}
\caption{\label{tab:ghz} The relevant portion of the GHZ-Mermin possibilistic empirical model. The suppressed rows of the table $\{XXY,XYX,YXX,YYY\}$ have full support. See Fig.~\ref{fig:ghz3d}~(a) for the three-dimensional representation of the model.}
\begin{center}
\begin{tabular}{p{3pt}p{3pt}l|cccccccc}
&~&~&  $000$ & $001$ & $010$ & $011$  & $100$ & $101$ & $110$ & $111$  \\ \hline
$X$&$X$&$X$ & $1$ & $0$ & $0$ & $1$ & $0$ & $1$ & $1$ & $0$  \\
$X$&$Y$&$Y$ & $0$ & $1$ & $1$ & $0$ & $1$ & $0$ & $0$ & $1$  \\
$Y$&$X$&$X$ & $0$ & $1$ & $1$ & $0$ & $1$ & $0$ & $0$ & $1$  \\
$Y$&$Y$&$X$ & $0$ & $1$ & $1$ & $0$ & $1$ & $0$ & $0$ & $1$ 
\end{tabular}
\end{center}
\end{table}

\drawSmallNKLCubes{	
	\pointAt{1}{0}{3}{red!50}
	\pointAt{0}{1}{3}{red!50}
	\pointAt{0}{0}{2}{red!50}
	\pointAt{1}{1}{2}{red!50}
	
	\pointAt{3}{0}{0}{red!50}
	\pointAt{2}{1}{0}{red!50}
	\pointAt{2}{0}{1}{red!50}
	\pointAt{3}{1}{1}{red!50}
	
	\pointAt{0}{2}{1}{red!50}
	\pointAt{1}{3}{1}{red!50}
	\pointAt{0}{3}{0}{red!50}
	\pointAt{1}{2}{0}{red!50}
	
	\pointAt{2}{2}{3}{red!50}
	\pointAt{3}{3}{3}{red!50}
	\pointAt{2}{3}{2}{red!50}
	\pointAt{3}{2}{2}{red!50}
}{
	\pointAt{3}{3}{0}{blue!50}
	
	\pointAt{1}{1}{2}{red!50}
	
	\pointAt{3}{0}{0}{red!50}
	
	\pointAt{0}{3}{0}{red!50}
	
	\pointAt{3}{3}{3}{red!50}
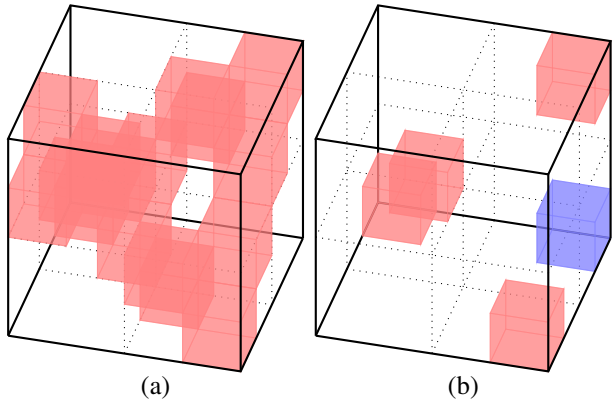
}{\label{fig:ghz3d} (a) The GHZ model. We represent only the red, impossible outcomes; all other entries are possible. (b) Hardy's paradox within the GHZ model; the blue outcome is possible.}

\begin{proposition}\label{prop:hghz}
The GHZ model witnesses Hardy nonlocality with certainty.
\end{proposition}

\begin{proof}
The three-dimensional representation makes it easy to identify a tripartite Hardy paradox, which is shown in Fig.~\ref{fig:ghz3d}~(b). It can also be expressed as follows.
\vbox{
\begin{itemize}
\item
$p(\,1,1,1\,\mid\,Y,Y,Y\,) > 0$
\item
$p(\,1,1,0\,\mid\,Y,Y,X\,) = 0$ \\
$p(\,1,0,1\,\mid\,Y,X,Y\,) = 0$ \\
$p(\,0,1,1\,\mid\,X,Y,Y\,) = 0$
\item
$p(\,0,0,0\,\mid\,X,X,X\,) = 0$
\end{itemize}
}
Up to re-labelling, this is the form of the $n$-partite Hardy paradox we met in Sec.~\ref{sec:background}.
Moreover, it can similarly be demonstrated that any joint outcome for the measurement context $YYY$ witnesses a Hardy paradox (this may be seen by inspection, but a detailed and more general treatment can also be found in the proof of Proposition~\ref{prop:ghznhardy} in the appendix to this article).
The paradoxical probability is
\[
p_\text{paradox} = p(1,1,1\mid Y,Y,Y) =0.125\, .
\]
However, since every outcome to the measurement $YYY$ witnesses some Hardy paradox, then it is again the case that the combined probability of witnessing Hardy nonlocality is $1$.
\end{proof}

This provides a much simpler tripartite Hardy argument than that of Ghosh, Kar and Sarkar \cite{ghosh:98}, using a simpler empirical model (theirs also used the GHZ state, but with alternative measurements on this state), while still obtaining the same value of $0.125$ for the individual paradoxical probabilities. Again, perhaps more importantly, in our model every possible outcome event for the joint measurement $YYY$ witnesses some Hardy paradox, and therefore Hardy nonlocality is witnessed with certainty. The model considered here is exactly the GHZ-Mermin model, given that the observables available at each subsystem are simply the $X$ and $Y$ operators. As a result, it can be said that the GHZ experiment \cite{greenberger:90} witnesses Hardy nonlocality with certainty.

\begin{corollary}
The GHZ experiment \cite{greenberger:90} witnesses Hardy nonlocality with certainty.
\end{corollary}

Mermin gave logical nonlocality proofs for $n$-partite generalisations of the GHZ state \cite{mermin:90b} for all $n>2$. Again, his arguments were not of the Hardy form, but we can generalise Proposition \ref{prop:hghz} to some of the GHZ($n$) models (see the appendix).
%
%
%
%

\section{Measurement Incompatibility Is Sufficient For Logical Nonlocality}

In \cite{wolf:09} it was shown that a pair of observables are incompatible in the sense of \emph{not} being jointly observable if and only if they enable a Bell inequality violation. Subsequent works have also considered how the degree of incompatibility relates to the degree of nonlocality \cite{oppenheim:10,banik:13}. Here, we show that, in the basic case of projective or sharp measurements, incompatibility is necessary and sufficient for logical nonlocality \footnote{The idea for Proposition~\ref{prop:incomp} arose from discussions with Leon Loveridge. Analysis of the unsharp case will be considered in future work.}.

\begin{proposition}\label{prop:incomp}
A pair of projective measurements enables a logical nonlocality argument if and only if it is incompatible.
\end{proposition}

\begin{proof}
In \cite{hardy:93}, it was shown that any non-maximally entangled two-qubit pure state can be written in the form
\begin{equation}\label{state}
\ket{\Psi} = N \left( -\alpha^* \beta^* \ket{u u_\perp} -\alpha^* \beta^* \ket{u_\perp u} + \alpha^2 \ket{u_\perp u_\perp}\right)
\end{equation}
for some orthonormal basis $\{ \ket{u}, \ket{u_\perp} \}$ and complex $\alpha,\beta$ such that $\alpha^2 + \beta^2 = 1$ and $\alpha^2>0$, where $N$ is simply a normalisation factor. Logical nonlocality in the form of the Hardy paradox is realised by local projective measurements on each qubit in the directions $\ket{u}$ and $\ket{d} := \alpha \ket{u} + \beta \ket{u_\perp}$.

A pair of non-commuting Hermitian operators has at least one pair of non-commuting spectral projections, say $P=\ketbra{u}{u}$ and $Q=\ketbra{d}{d}$ for some $\ket{u}$ and $\ket{d}$. For the moment let us not assume any relation to the vectors considered in the previous paragraph. The projections are used to build a pair of non-commuting two-outcome observables $\tilde{P}:=2P-\mathbb{1}$ and $\tilde{Q}:=2Q-\mathbb{1}$. Essentially, these correspond to course-graining the probabilities of all outcomes not corresponding to $\ket{u}$ or $\ket{d}$, respectively.
We may assume that $\ket{d} = \alpha \ket{u} + \beta \ket{u_\perp}$ for some $\ket{u_\perp}$ orthogonal to $\ket{u}$ and complex $\alpha,\beta$ such that $\alpha^2+\beta^2=1$ and $\alpha^2 > 0$, for otherwise the projections $P$ and $Q$ would commute. Now suppose we have a bipartite system in which each party may choose to measure $P$ or $Q$. Having defined $\ket{u}$ and $\ket{u_\perp}$ in this way, logical nonlocality in the form of a coarse-grained Hardy paradox is realised on the entangled state specified by Eq.~(\ref{state}).
\end{proof}

\section{Conclusion}


Theorem~\ref{equivprop}, which combines the completeness results proved by the author and Fritz in \cite{mansfield:11}, has been seen in this article to lead to an abundance of consequences and applications which we now briefly recap.

The polynomial algorithms for deciding logical nonlocality in $(2,2,l)$ and $(2,k,2)$ scenarios given in Sec.~\ref{sec:complexity} are of particular relevance since the problem is known in general to be NP-complete \cite{abramsky:13b}. Further scenarios have been shown to be tractable elsewhere \cite{simmons2016computational}.

It was already known that PR boxes are the only strongly nonlocal $(2,2,2)$ models \cite{abramsky:11,lal:13}, but the proof obtained in Sec.~\ref{sec:pr} provides more insight than the previously existing computational proof: in particular it is seen that the result is a straighforward consequence of the completeness of Hardy nonlocality and the property of no-signalling.

Given that all $n$-partite entangled states admit logical nonlocality proofs via projective measurements for $n>2$ \cite{abramsky:16hardy}, Theorem~\ref{thm:ba} establishes the rather surprising fact that in this respect the Bell states are uniquely anomalous in the landscape of entangled states \footnote{Moreover, using an independent result due to Brassard and M\'ethot \cite{brassard:13} it may be deduced that the anomaly persists in the case that POVMs are permitted.}.

The paradoxical probability has often been viewed as an indicator of the quality of Hardy and logical nonlocality, and the issue of optimising this figure in various systems has been widely considered in the literature. The results of Sec.s~\ref{sec:chen} and \ref{sec:hpcert} provide a clarifying perspective on this issue. In particular it is seen that certain logical nonlocality proofs which claim to achieve high paradoxical probabilities are rather summing the paradoxical probabilities of numerous Hardy paradoxes which are present. It may indeed be argued that this total probability of witnessing logical nonlocality is a more relevant indicator and potential measure of the quality of logical nonlocality.
If we accept this, it casts the issue of optimisation in a rather new light, since we have seen that Hardy nonlocality can be achieved with certainty in a tripartite quantum system: something which in fact is already verified by the GHZ experiment. Moreover, the property of witnessing logical nonlocality \emph{with certainty} is understood to be equivalent to the property of strong nonlocality.

While previous works have considered how measurement incompatibility relates to nonlocality in terms of Bell inequality violations, Proposition~\ref{prop:incomp} provides initial progress on the question of how incompatibility relates to other classes of nonlocality in the qualitative hierarchy, which will be a topic for future research.

As a final open question, we note that a correspondence has been established between possibilistic empirical models and relational database theory~\cite{abramsky:12c}. It remains to be explored whether Theorem~\ref{equivprop} might find applications in database theory, or indeed whether similar results already exist in that field that might lead to further insights in the study of nonlocality.

\section*{Acknowledgements}
The author thanks Samson Abramsky, Rui Soares Barbosa, Tobias Fritz, Lucien Hardy, Ray Lal, Leon Loveridge, Andrew Simmons and Jamie Vicary for comments and discussions, Johan Paulsson for invaluable help with figures, and gratefully acknowledges financial support from the Fondation Sciences Math\'{e}matiques de Paris, postdoctoral research grant eotpFIELD15RPOMT-FSMP1, Contextual Semantics for Quantum Theory. This work was partially carried out at l'Institut de Recherche en Informatique Fondamentale, Universit\'e Paris Diderot - Paris 7, the Simons Institute, University of California, Berkeley as the Logical Structures in Computation programme, and the Department of Computer Science, University of Oxford.


\section*{Appendix: GHZ($n$)}

Mermin gave logical nonlocality proofs for $n$-partite generalisations of the GHZ state \cite{mermin:90b} for all $n>2$. These arguments are not of the Hardy form, but we will now show how to generalise Proposition \ref{prop:hghz} to some of the GHZ($n$) models.

The GHZ($n$) states are:
\begin{equation}\label{eq:ghznstate}
\ket{\text{GHZ}(n)} := \frac{1}{\sqrt{2}} \left( \ket{0 \cdots 0} + \ket{1\cdots 1}\right),
\end{equation}
where $n$ is the number of qubits. Note that for $n=2$ the state obtained would be the $\ket{\phi^+}$ Bell state. For $n>2$, Mermin considered models in which each each party can make Pauli $X$ or $Y$ measurements. With a little calculation, it is possible to concisely describe the resulting empirical models in a logical form \footnote{The calculations which follow are based on notes from a private communication with Samson Abramsky}.

The eigenvectors of the $X$ operator are
\begin{equation}\label{eq:evx}
\ket{0_x} = \frac{1}{\sqrt{2}} \left( \ket{0} + e^{i0}\ket{1} \right), \qquad \ket{1_x} = \frac{1}{\sqrt{2}} \left( \ket{0} + e^{i \pi}\ket{1} \right).
\end{equation}
The vector $\ket{0_x}$ has eigenvalue $+1$ and the vector $\ket{1_x}$ has eigenvalue $-1$. These are more usually denoted $\ket{+}$ and $\ket{-}$, respectively, but we use an alternative notation to agree with the $\{0,1\}$ labelling of outcomes used in this article. The phases have been made explicit since they will play the crucial role in the following calculations. Similarly, the $+1$ and $-1$ eigenvectors of the $Y$ operator are
\begin{equation}\label{eq:evy}
\ket{0_y} = \frac{1}{\sqrt{2}} \left( \ket{0} + e^{i \pi/2}\ket{1} \right), \quad \ket{1_y} = \frac{1}{\sqrt{2}} \left( \ket{0} + e^{-i \pi/2}\ket{1} \right).
\end{equation}

The various probabilities for these quantum-mechanical empirical models can be calculated as
\[
\left| \braket{\text{GHZ}(n)}{v_1 \dots v_n} \right|^2 ,
\]
where the $v_i$ are the appropriate eigenvectors. This evaluates to
\begin{equation}\label{eq:ghznphase}
\left| \frac{1+ e^{i \phi}}{\sqrt{2^{n+1}}} \right|^2 = \frac{1}{2^n}\left( 1 + \cos \phi \right),
\end{equation}
where $\phi$ is the sum of the phases of the $v_i$. From the phases of the possible eigenvectors, (\ref{eq:evx}) and (\ref{eq:evy}), it is clear that we must have $\phi = k \,\pi/2$ for some $k \in \mathbb{Z}_4$, the four element cyclic group. For $k = 0 \bmod{4}$, the probability will be $\frac{1}{\sqrt{2^{n-1}}}$; for $k = 1$ or $3 \bmod{4}$ the probability will be $\frac{1}{\sqrt{2^n}}$; and for $k = 2 \bmod{4}$ the probability will be zero.

We can now reduce the calculation of probabilities for any such model to a simple counting argument. If $k_{0_x}$ is the number of $\ket{0_x}$ eigenvectors, $k_{1_x}$ is the number of $\ket{1_x}$ eigenvectors, and so on, then
\begin{align*}
k &= k_{0_y} + 2 \cdot k_{1_x} + 3 \cdot  k_{1_y}  \pmod{4} \\
&= \left( k_{0_y} +k_{1_y} \right) + 2\cdot\left( k_{1_x} +k_{1_y} \right) \pmod{4}.
\end{align*}

\begin{itemize}
\item
For contexts containing an odd number of $Y$'s, every outcome is possible with equal probability $\frac{1}{\sqrt{2^n}}$, since $k = 1$ or $3 \bmod{4}$.
\item
For contexts containing $0 \bmod{4}$ $Y$'s, outcomes are possible if and only if they contain an even number of $1$'s. For these outcomes, $k = 0 \bmod{4}$ and the probabilities are $\frac{1}{\sqrt{2^{n-1}}}$. If there were an odd number of $0$'s in the outcome then $k = 2 \bmod{4}$ and the probability would be $0$.
\item
Similarly, for contexts that contain $2\text{ mod }4$ $Y$'s, outcomes are possible if and only if they contain an odd number of $1$'s. Again, the non-zero probabilities are $\frac{1}{\sqrt{2^{n-1}}}$.
\end{itemize}

Though the probabilities are seen to be easily be calculated in this way, we need only concern ourselves with the possibilistic information in what follows.

\begin{proposition}\label{prop:ghznhardy}
All GHZ($n$) models for $n = 3 \bmod 4$ witness an $n$-partite Hardy paradox with certainty.
\end{proposition}

\begin{proof}
Proposition \ref{prop:hghz} showed that this holds for $n=3$. Let $\mathbf{o} = (o_1, \dots, o_n)$ be any binary string of length $n$, let $\gamma_i$ be the function that changes the $i$th entry of a binary string, and let $\mathbf{o}^{-1}$ denote the binary string of length $n$ which differs in every entry from $\mathbf{o}$. We show that every outcome $\mathbf{o}$ to the measurements $(Y,\dots,Y)$ witnesses a Hardy paradox. We deal with the cases that $\mathbf{o}$ has an even or an odd number of $1$'s separately.

Suppose $\mathbf{o}$ has an even number of $1$'s:
$p\left(\mathbf{o}\mid Y, \dots, Y\right) >0$,
since there are an odd number of $Y$ measurements;
$p\left(\,\mathbf{o} \,\mid\, \pi(X,Y,\dots,Y) \,\right) = 0$, for all permutations $\pi$, since there are $2 \bmod 4$ $Y$'s and $\mathbf{o}$ has an even number of $1$'s;
$p\left(\mathbf{o}^{-1}\mid X, \dots, X \right) = 0$, since there are $0 \bmod 4$ $Y$'s and $\mathbf{o}^{-1}$ has an odd number of $1$'s.

Suppose $\mathbf{o}$ has an odd number of $1$'s:
$p\left(\mathbf{o}\mid Y, \dots, Y\right) >0$,
since there are an odd number of $Y$ measurements;
$p\left(\,\gamma_i(\mathbf{o})\,\mid\,\gamma_i(Y,Y,\dots,Y) \,\right) = 0$, for all permutations $i = 1,\dots, n$, since there are $2 \bmod 4$ $Y$'s and an even number of $1$'s in $\gamma_i (\mathbf{o})$;
$p\left(\mathbf{o}\mid X, \dots, X \right) = 0$, since there are $0 \bmod 4$ $Y$'s and an odd number of $1$'s in $\mathbf{o}$.
\end{proof}

It should be pointed out that even though we can say that Hardy nonlocality will can be witnessed wtih certainty in all of these models, the paradoxical probabilities for the individual Hardy paradoxes are always $\sfrac{1}{2^n}$, with the maximum obtained for the tripartite GHZ model.

Such a result does not hold for GHZ($n$) models for which $n \neq 3 \bmod 4$, as it can be seen that these models do not contain $n$-partite Hardy paradoxes. This follows from the fact that any $(n,2,2)$ Hardy paradox must take the form of one of the paradoxes in the proof of Proposition \ref{prop:ghznhardy}, but it can easily be checked that the counting arguments for identifying such paradoxes in GHZ($n$) models work if and only if $n=3 \bmod 4$.

\bibliographystyle{unsrt}
\bibliography{refs1}

\end{document}